\documentclass[runningheads]{llncs}
\usepackage[T1]{fontenc}
\usepackage{graphicx}
\usepackage{amsmath}
\usepackage{amssymb}

\newtheorem{assumption}{Assumption}
\begin{document}
\title{Q-GARS: Quantum-inspired Robust Microservice Chaining Scheduling}
%
%
\author{Huixiang Zhang\inst{1}\orcidID{0009-0006-5250-2442} \and
Mahzabeen Emu\inst{2,3}\orcidID{0000-0002-0433-1873}}
\authorrunning{H. Zhang et al.}
%
\institute{Lakehead University, Thunder Bay Ontario P7B 5E1, Canada \and
Memorial University, St. John's Newfoundland and Labrador, A1C 5S7 Canada \and
Quantum Communications and Computing Center, St. John's Newfoundland and Labrador, A1C 5S7 Canada}
\maketitle              
\begin{abstract}
Microservice-based applications are characterized by stochastic latencies arising from long-tail execution patterns and heterogeneous resource constraints across computational nodes. To address this challenge, we first formulate the problem using Quadratic Unconstrained Binary Optimization (QUBO), which aligns the problem with emerging quantum-optimization paradigms. Building upon this, we propose Q-GARS (Quantum-Guided Adaptive Robust Scheduling), a hybrid framework that integrates the QUBO model with Simulated Quantum Annealing (SQA) based combinatorial search and online rescheduling mechanisms, enabling global microservice rank generation and real-time robust adjustment. We treat the SQA-produced rank as a soft prior, and update a closed-loop trust weight to adaptively switch and mix between this prior and a robust proportional-fairness allocator, maintaining robustness under prediction failures and runtime disturbances. Simulation results demonstrate that Q-GARS achieves an average weighted completion time improvement of 2.1\% relative to a shortest-remaining-processing-time (SRPT) greedy baseline, with performance gains reaching up to 16.8\% in heavy-tailed latency. The adaptive mechanism reduces tail latency under high-variance conditions. In addition, Q-GARS achieves a mean node resource utilization rate of 0.817, which is 1.1 percentage points above the robust baseline (0.806).

\end{abstract}
\section{Introduction}

As the computation capability of personal or edge devices grows, some latency-sensitive microservices can be deployed on edge nodes close to the user to improve the quality of service (QoS) \cite{fu2022}. For interactive online services, microservices are typically composed as directed acyclic graphs \cite{hu2023}. These graphs process end-to-end requests. In this architecture, performance bottlenecks often focus on node-local runtime scheduling \cite{li2014}. Concurrent microservice instances from multiple workflows compete for multidimensional computing resources on the same physical node. This is a key driver of end-to-end tail latency and service-level objectives.

In the cloud edge continuum, microservice execution times are highly stochastic \cite{li2014}. Input fluctuations and multi-tenant interference can create heavy-tailed latency. This can trigger head-of-line blocking. Meanwhile, service meshes provide traffic management and observability. However, their data-plane proxies may also introduce extra overhead and amplify latency variability \cite{zhu2023}. Therefore, a scheduler must operate within millisecond-scale decision windows. The scheduler needs to handle discrete priority decisions and continuous multi-resource allocation constraints within millisecond-scale budgets, while remaining robust to prediction errors.

Fig. \ref{fig:problem_illustration} depicts a microservice-based user-facing application in the cloud–edge continuum, highlighting the end-to-end request paths and node-local scheduling bottlenecks. Requests arrive at the edge ingress, where IAM authentication is performed followed by a cache lookup. On a cache hit, the response is returned directly from the edge cache (fast path). On a cache miss, the request is enqueued to the edge OCR microservice for text extraction, and then forwarded to the compute-intensive inference pipeline, which is typically offloaded to a pool of cloud-side inference replicas. These stages create two contention points—the edge OCR queue and the cloud inference queue—where online scheduling decisions (ranking and rate allocation) directly shape end-to-end tail latency. In particular, heavy-tailed service times and volatile arrivals (exacerbated by cache hit/miss traffic splitting) can trigger head-of-line blocking, so that a few long-tail tasks amplify queuing delays even when aggregate resources are sufficient, leading to QoS/SLA violations under bursts.

\begin{figure}[htbp]
    \centering
    \includegraphics[width=\linewidth]{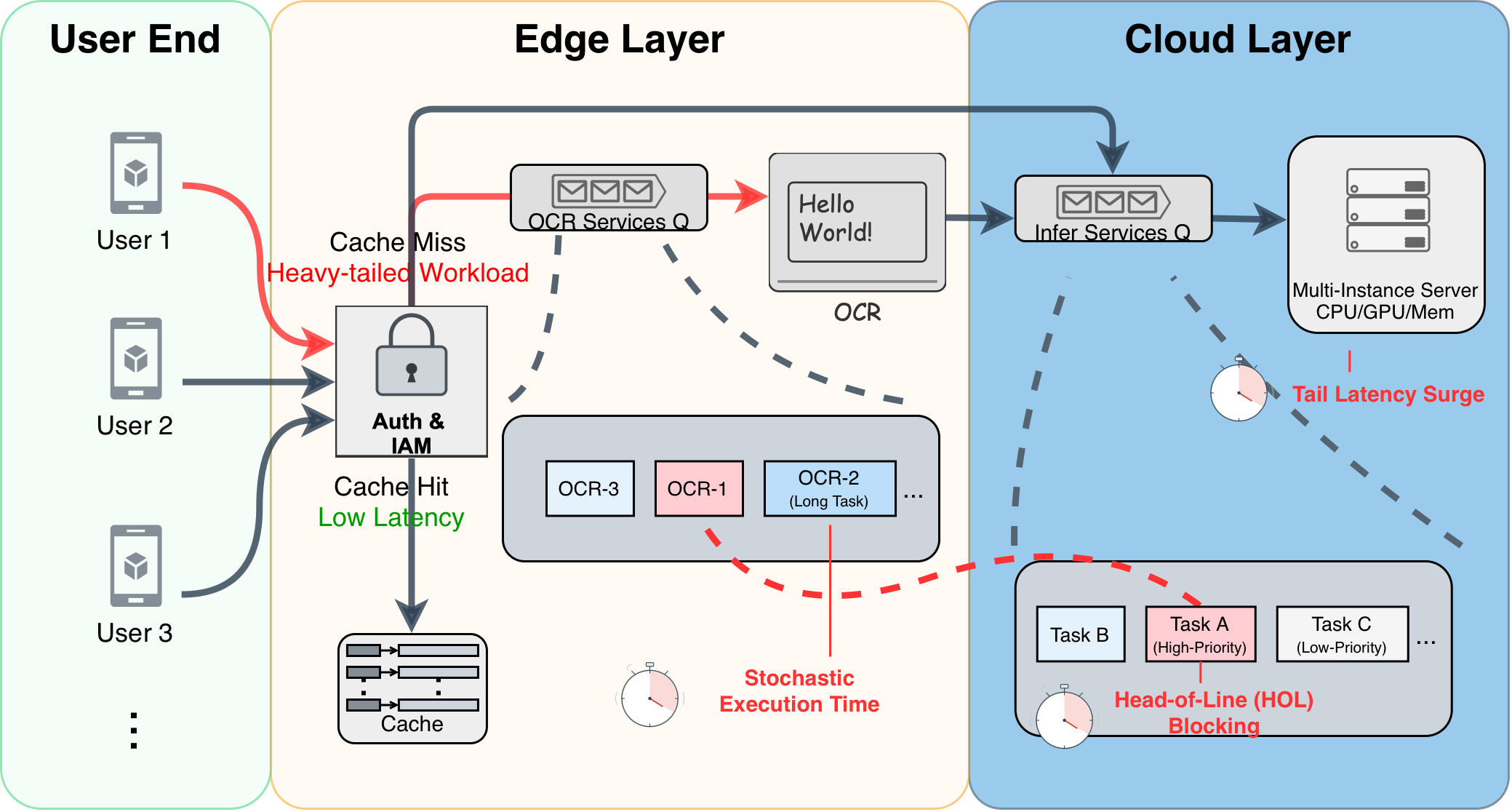} 
    \caption{End-to-end request paths and node-local scheduling bottlenecks in a microservice-based user-facing application.}
    \label{fig:problem_illustration}
\end{figure}

Consequently, an effective node-level scheduler must simultaneously execute discrete prioritization (to bypass HOL blocking) and continuous multidimensional rate allocation (to manage resource contention), all within a millisecond-scale online budget. To tractably solve this daunting mixed-integer optimization challenge, we propose a Quantum-Inspired and Adaptive Robust Scheduling framework (Q-GARS). The key idea is to decompose the node-level scheduling decision into two steps. First, we encode the priority ranking of ready microservice instances as a quadratic unconstrained binary optimization (QUBO) model. We use an annealing-type metaheuristic algorithm to generate a high-quality discrete ranking prior under strict time budgets. Subsequently, conditioned on the ranking, we compute continuous execution rates by network utility maximization (NUM) \cite{kelly1998}. This yields a resource slicing strategy that satisfies multidimensional capacity constraints.

Because the ranking prior relies on runtime estimates, predictions may fail under severe volatility. We adopt an adaptive robust execution mechanism. The system treats the previous policy as an untrustworthy expert \cite{purohit2018,lindermayr2025}. It competes in parallel with a stable robust baseline policy. The system adjusts the trust in the prior online via an exponential weight update. The scheduler ultimately executes an adaptive convex combination of the two. This design preserves empirical performance gains when the predictions are accurate. It provides a provable performance guaranty relative to the baseline when predictions degrade.

We systematically evaluated the framework on a discrete-event simulation platform. We construct controllable heterogeneous resource constraints and stochastic latency injections. This tests system stability and graceful degradation behavior under heavy-tailed latency and telemetry disturbances. The main contributions of this paper are summarized as follows:

1. Decoupled node-level scheduling architecture: We map QUBO-based discrete rankings to continuous execution rates satisfying resource constraints via NUM and dynamic penalty calculation.

2. Robust adaptive execution mechanism: We provide a sublinear regret guarantee relative to the robust baseline when mixing the prior and the baseline via exponential weights, yielding a worst-case safeguard under non-stationary shocks. This mechanism limits the 95th percentile queue backlog peak to 20\%-30\% of the baseline.

3. Through discrete-event simulations at scale (4096 parallel runs), the framework reduces the average weighted completion time by 2.1\%, with improvements of up to 16.8\% in complex topologies. In high-volatility scenarios, the system effectively suppresses tail latency and ensures rapid recovery.

The rest of the paper is organized as follows. Section 2 reviews related work. Section 3 presents the system model and the global optimization objective. Section 4 presents the node-level QUBO formulation and the implementation of the adaptive mechanism. Section 5 reports on the experimental setup and key results. Section 6 concludes the paper with limitations and future plans.

\section{Related Work}

In the traditional operations research domain, microservice scheduling and service chaining problems are often formulated as Integer Linear Programming (ILP) or Mixed-Integer Linear Programming (MILP) models. 
For example, Harutyunyan et al. formulate a joint user association, service function chain placement, and resource allocation problem using ILP, followed by a heuristic to address scalability \cite{harutyunyan2019}. Similarly, Kiji et al. formulate multicast service chaining as an ILP that jointly decides VNF placement and routing to minimize deployment and link-usage costs \cite{kiji2019}. 
Beyond service chaining, ILP is also a standard tool for delay-sensitive distributed server allocation and fault-tolerant provisioning, where exact formulations are often accompanied by approximation or heuristic methods for efficiency \cite{kawabata2023,inoue2026}. 
Such formulations provide a principled way to encode global coupling constraints and yield optimal solutions when solved to optimality, but they often require heuristics, relaxations, or approximation algorithms to meet practical scalability and online time budgets.

To overcome the computational bottleneck of exact ILP/MILP solvers, researchers have begun adopting emerging computing paradigms such as Quantum Annealing (QA) and its classical counterpart, Simulated Quantum Annealing (SQA). Zhang et al. further map online prioritization decisions to a rank-based QUBO formulation, enabling fast permutation search under tight decision windows \cite{zhang2022}. 
Although many QA scheduling demonstrations are conducted on canonical OR benchmarks (e.g., JSP/RCPSP), these models capture core structures that also arise in microservice workflow scheduling: precedence constraints induced by service chains and shared (multi-)resource capacity constraints at execution sites. 
For example, Pérez Armas et al. study the resource-constrained project scheduling problem (RCPSP), systematically compare multiple MILP formulations, convert a qubit-efficient formulation into a QUBO with explicit penalty terms for hard constraints, and evaluate it on the D-Wave Advantage 6.3 quantum annealer \cite{perezarmas2024}. This line of work supports using QUBO/QA as a time-budgeted combinatorial search primitive for resource-constrained scheduling, an abstraction that naturally aligns with node-level microservice scheduling in the cloud-edge continuum.

From a cloud systems perspective, recent microservice scheduling framework emphasizes end-to-end QoS under communication overhead and multi-tenant resource contention. Fu et al. propose Nautilus, which combines a communication-aware microservice mapper with contention-aware per-node resource management and runtime migration mechanisms to improve resource efficiency while protecting tail-latency objectives under interference \cite{fu2022}. These systems demonstrate the necessity of fast, interference-aware control loops inside the cloud-edge continuum. Their node-level controllers typically rely on learned policies or heuristic rules rather than an explicit, time-budgeted combinatorial search primitive that can rapidly explore priority permutations under strict millisecond deadlines. 
This motivates our decoupled design that uses QUBO/SQA to generate a lightweight structural prior for ordering decisions, while retaining continuous multi-resource allocation and robust safeguarding to handle stochastic latencies.

\section{System Model}
In this section, we formalize the microservice scheduling problem within the cloud-edge architecture as a discrete-time control model. We adopt a two-timescale architecture to manage the system. The global orchestrator handles service placement logic over coarse-grained periods. The node-level scheduler executes local control at millisecond-scale decision epochs. At each discrete decision epoch, the system observes and extracts the set of ready tasks that have satisfied their upstream directed acyclic graph dependencies as the current state. The control action of the system is to compute the instantaneous resource allocation rates subject to the multidimensional capacity limits of the heterogeneous nodes. The global optimization objective of the system is to minimize the total weighted completion time of all microservice workflows. Given the large scale of active workflows and microservices, the global objective function is computationally intractable to solve directly within a single decision epoch. The node-level controller utilizes local queueing delay and resource contention penalties as a tractable surrogate for this global objective to make real-time decisions. Table \ref{tab:notations} summarizes the core mathematical notations used throughout this paper.

\begin{table}[h]
    \centering
    \caption{Summary of Core Sets, States, and Decision Variables}
    \label{tab:notations}
    \renewcommand{\arraystretch}{1.3}
    \begin{tabular}{|l|l|}
        \hline
        \multicolumn{2}{|c|}{Sets and System Parameters} \\
        \hline
        Notation & Description \\
        \hline
        $\mathcal{J}$ & Set of user request workflows, indexed by $j$ \\
        $\mathcal{M}$ & Set of heterogeneous compute nodes, indexed by $m$ \\
        $C_m$ & Multidimensional computing resource capacity vector $C_m \in \mathbb{R}_+^d$ \\
        $\omega_j$ & Service level agreement weight of workflow $j$ \\
        \hline
        \multicolumn{2}{|c|}{State and Control Variables} \\
        \hline
        Notation & Description \\
        \hline
        $t$ & Discrete decision epoch \\
        $S_m(t)$ & Node-local active ready microservice set extracted at decision epoch $t$ \\
        $\mathbf{r}_v(t)$ & Multidimensional execution rate vector allocated to microservice $v$ \\
        $\mathbf{r}^*(t)$ & Final deterministic mixed execution rate vector \\
        $\lambda(t)$ & Adaptive trust parameter for the online learning framework \\
        \hline
        \multicolumn{2}{|c|}{Intermediate Algorithmic Variables (Section 4)} \\
        \hline
        Notation & Description \\
        \hline
        $r$ & Discrete execution rank index for microservices \\
        $c_{v,r}$ & Base linear proxy cost of assigning microservice $v$ to rank $r$ \\
        $\theta_v(t)$ & Continuous computing resource allocation share for microservice $v$ \\
        \hline
    \end{tabular}
\end{table}

\subsection{Infrastructure and System State}

We model the physical cloud-edge infrastructure as a set of heterogeneous compute nodes $\mathcal{M}$. Each node $m \in \mathcal{M}$ possesses a multidimensional resource capacity vector $C_m \in \mathbb{R}_+^d$. This vector covers hardware limits such as CPU, memory, and network bandwidth. These limits mathematically define a multidimensional packing polytope. This establishes the physical boundary for local node scheduling. User service requests are represented as a set of workflows $\mathcal{J}$. Each workflow $j \in \mathcal{J}$ is mapped to a directed acyclic graph $G_j = (V_j, E_j)$. The vertices $v \in V_j$ represent specific microservice tasks. The directed edges $(u,v) \in E_j$ represent strict upstream data dependencies between tasks. We denote the execution node assigned to task $v$ by the global orchestrator as $\mu(j,v) \in \mathcal{M}$. Because the actual execution latencies of microservices are non-clairvoyant, the node-level scheduler acts solely based on the observable state at discrete decision epochs $t$. A key component of the observable local state is the node-local ready set $S_m(t)$. The system telemetry state simultaneously includes the queue backlog $q_v(t)$ and the observed enqueue rate $a_v(t)$ for task $v$. Task $v$ is included in $S_m(t)$ and becomes eligible to compete for available intra-node resources only when all its upstream predecessor tasks $u$ have completed and $\mu(j,v) = m$.

\subsection{Control Action and Surrogate Objective}
After we defined the system state, we further specify the control action and the optimization objective. At each decision epoch $t$, the scheduler on node $m$ allocates an instantaneous multidimensional execution rate vector $\mathbf{r}_v(t) \in \mathbb{R}_+^d$ to each task $v$ in the local active ready set $S_m(t)$. The aggregate resource allocation for all concurrent tasks must satisfy the multidimensional packing polytope constraint. We mathematically define this multidimensional packing polytope as a time-dependent set of decision vectors. Its form is $\mathcal{P}_m(t) = \{ \{ \mathbf{r}_v(t) \}_{v \in S_m(t)} : \mathbf{r}_v(t) \ge 0, \sum_{v \in S_m(t)} \mathbf{r}_v(t) \le C_m \}$. The global optimization objective of the system is to minimize the total weighted completion time of all workflows. This objective is formulated as $\min \sum_{j \in \mathcal{J}} \omega_j C_j$. Here $C_j$ represents the end-to-end completion time of the final microservice task in workflow $j$. Because of the non-clairvoyant nature of the real physical environment, the actual execution latency $p_v$ of a microservice is stochastic and unknown. The system can only observe this latency upon task completion. This makes the global objective function intractable to solve directly within a single decision epoch. To obtain an implementable online scheduler, we approximate the global weighted-completion-time objective with a node-local surrogate that can be evaluated at each millisecond-scale decision epoch.
Intuitively, the scheduler should (i) prioritize tasks with large backlogs (to reduce queueing delay and tail latency), and (ii) avoid co-running strongly interfering tasks (to mitigate resource contention).
Accordingly, we use a surrogate objective of the form
\begin{equation}\label{eq:obj}
\min \sum_{v \in S_m(t)} \phi\!\left(q_v(t), \mathbf r_v(t)\right)\;+\!\!\sum_{v<u}\psi\!\left(\mathbf r_v(t), \mathbf r_u(t)\right),
\end{equation}
where $\phi$ quantifies per-task delay pressure.
We instantiate it as
\begin{equation}\label{eq:phi}
  \phi(q_v, r_v) \;=\; \frac{q_v}{\theta_v},
\end{equation}
where $\theta_v$ is the scalar resource share defined in Section~4.3. This function is strictly decreasing and strictly convex in $\theta_v$
for $\theta_v > 0$; this convexity property is used in the theoretical guarantee of Section~4.4. The term $\psi$ quantifies pairwise contention for shared resources.
These two components serve different architectural phases of the framework. In Section~4.2, $\phi$ and $\psi$ are jointly encoded into a permutation-based QUBO model: $\phi$ induces the linear ranking cost $c_{v,r}$ (urgency of placing $v$ at rank $r$), and $\psi$ induces the pairwise interference coefficient $Q_{v,u}$ modulated by a distance-decaying kernel $g(\cdot)$, enabling a QUBO encoding for fast permutation search under strict time budgets. In the adaptive execution phase of Section~4.4, the permutation has already been fixed and the contribution of $\psi$ has been absorbed into the discrete ranking result. The real-time shadow loss is therefore computed
from the aggregate $\phi$ values only, which preserves the convexity of the per-epoch cost with respect to the rate vector.

\section{Quantum-Inspired and Adaptive Robust Scheduling Framework}

This section presents an online scheduling framework to approximately solve the local surrogate objective introduced in Section 3 under millisecond-level decision budgets. Fig. \ref{fig:framework_illustration} illustrates the overall architecture of this framework. We use a local QA simulator driven by SQA algorithms. Because our framework models the scheduling problem as a standardized QUBO via the simulator, it is entirely decoupled from the underlying solver. This design allows our system to seamlessly switch to real physical QA devices once edge-quantum hardware matures or communication latency bottlenecks are resolved. To achieve this goal, we adopt a four-phase design: local state extraction and dimensionality reduction, quantum-inspired structural prior generation, continuous rate allocation, and adaptive robust execution.

\begin{figure}[htbp]
    \centering
    \includegraphics[width=\linewidth]{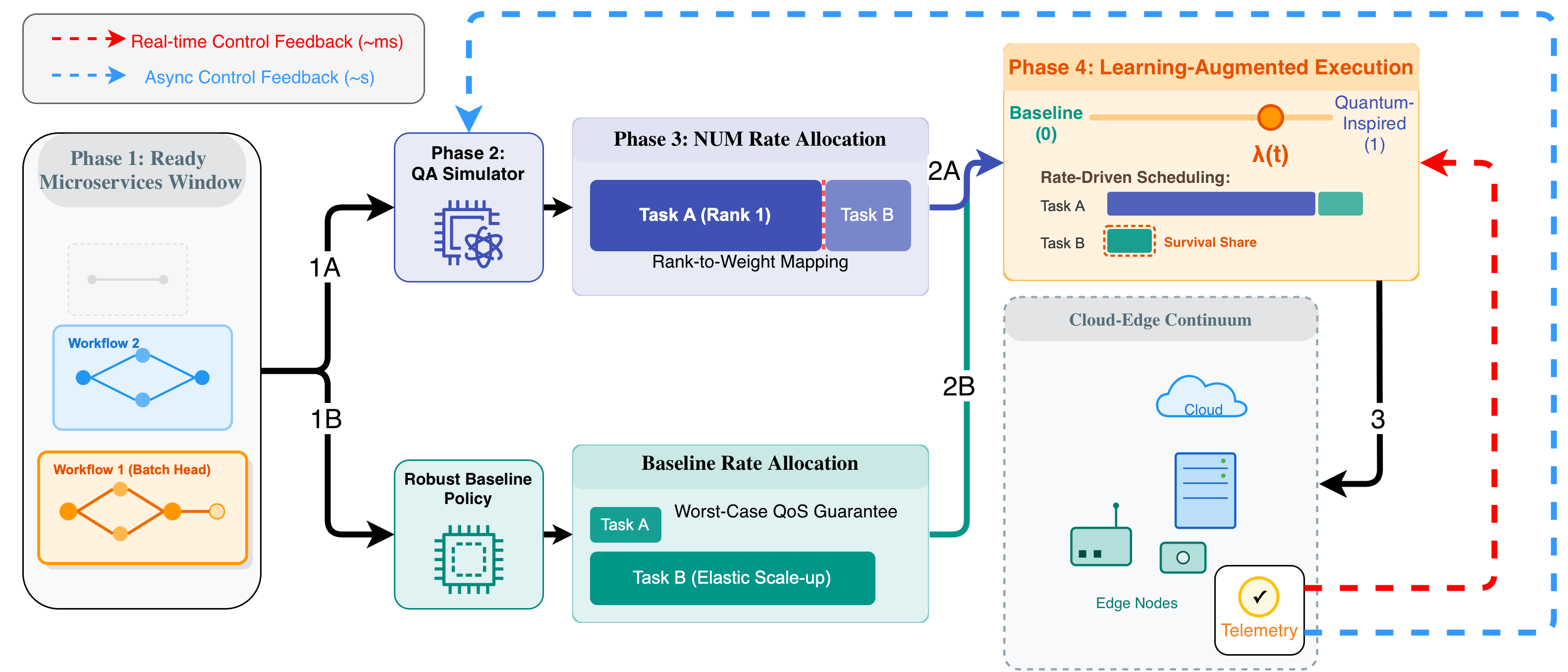} 
    \caption{Overview of the closed-loop adaptive scheduling framework integrating SQA and real-time feedback control.}
    \label{fig:framework_illustration}
\end{figure}

\subsection{State Extraction and Dimensionality Reduction}
In the first phase, at each decision epoch $t$, the controller at the node-level extracts the local active ready set of the node $S_m(t)$ from the acyclic graphs directed to the workflow, illustrated by the leftmost module in Fig. \ref{fig:framework_illustration}. Since microservices in $S_m(t)$ have satisfied all hard upstream dependencies, they are eligible to compete for intra-node computing resources. Directly optimizing the continuous multidimensional execution rate vector $\mathbf{r}_v(t)$ for all ready microservices is prohibitive. This is due to the combinatorial structure induced by contention and priority interactions. To limit the combinatorial search space, we retain only the top candidate microservices $K$ according to a lightweight priority score. We use the arrival order as a tie-breaker. This forms a dynamic active window $\tilde{S}_m(t)$. This yields $O(K^2)$ binary variables in the subsequent permutation-indexed QUBO encoding.

\subsection{QUBO Formulation and Rank Generation}
In the second phase, we map the microservice prioritization into a QUBO model. We define binary decision variables $x_{v,r} \in \{0,1\}$ for the active window $\tilde{S}_m(t)$. This variable equals 1 if microservice $v \in \tilde{S}_m(t)$ is assigned to the $r$-th execution rank. Otherwise, this variable equals 0. Here $1 \le r \le K$. The complete objective function consists of a cost part and a penalty part. Its formula is $H(\mathbf{x}) = H_{\text{cost}}(\mathbf{x}) + H_{\text{pen}}(\mathbf{x})$.

The cost function captures the linear delay penalties and soft non-linear computing resource contention among microservices:
\begin{equation}
H_{\text{cost}}(\mathbf{x}) = \sum_{v \in \tilde{S}_m(t)} \sum_{r=1}^{K} c_{v,r} x_{v,r} + \sum_{v, u \in \tilde{S}_m(t), v \neq u} \sum_{1 \leq s < r \leq K} g(r-s) Q_{v,u} x_{v,r} x_{u,s}   
\end{equation}

The term $c_{v,r}$ represents the base linear proxy cost of assigning a microservice $v$ to the rank $r$. It is induced by the delay-pressure function $\phi$ (Section~3.2) under the rank substitution and is defined as $c_{v,r} = q_v(t) \cdot r$, where $q_v(t)$ is the current queue backlog of $v$. This definition captures a direct scheduling intuition: tasks with larger backlogs incur higher cost when placed at lower priority rank positions. Because explicit concurrent execution overlap is difficult to model using pure ordinal variables, we introduce $Q_{v,u}$ as an order-induced interference proxy that captures the contention intensity between two tasks for shared resources. We define it as $Q_{v,u} = q_v(t)\, q_u(t) \,/\, (\max_{w \in \tilde{S}_m(t)} q_w(t))^2$, which ensures $Q_{v,u} \in [0,1]$. The distance decay kernel is defined as $g(\Delta) = \exp(-\beta\,\Delta)$, where $\beta > 0$ is a decay rate parameter and $\Delta = r - s \geq 1$. This kernel satisfies $g(\Delta) \in (0,1)$ and decreases with rank separation. Its physical interpretation is that the externality of resource contention imposed on a microservice $v$ at rank $r$ by a competing microservice $u$ at a higher-priority rank $s$ decays as the rank gap between them increases.

To convert the problem into an unconstrained form, the penalty function enforces a strict one-to-one mapping between tasks and ranks via quadratic terms:
$$
H_{\text{pen}}(\mathbf{x}) = A \sum_{v \in \tilde{S}_m(t)} \left(\sum_{r=1}^{K} x_{v,r} - 1\right)^2 + B \sum_{r=1}^{K} \left(\sum_{v \in \tilde{S}_m(t)} x_{v,r} - 1\right)^2
$$
$A$ and $B$ are penalty coefficients. Static empirical lower bounds often generate excessively large penalty weights. These weights overwhelm the underlying cost optimization objective and cause the solver to stall in local optima. To resolve this issue within millisecond budgets, we adopt a dynamic maximum marginal cost method to automatically set the penalty coefficients. Given that the distance-decaying kernel satisfies $0 \le g(r-s) \le 1$, we can safely omit $g(\cdot)$ to form a conservative upper bound for quadratic interactions. During the construction of the cost matrix, the system calculates the maximum local objective increment that can occur when any single microservice is assigned. We set the penalty coefficients as $A = B = \max_{v \in \tilde{S}_m(t)} ( \max_r c_{v,r} + \sum_{u \neq v} Q_{v,u} )$. This choice serves as a conservative sufficient condition. It ensures that the penalty incurred by any single-variable move violating the constraints cannot be offset by an improvement in the cost Hamiltonian $H_{\text{cost}}$. This dynamic calculation strictly bounds the time complexity to $O(K^2)$. It not only generates adaptive penalty bounds that ensure feasibility dominance but also avoids the additional latency imposed on the control loop by complex matrix parsing operations.

\subsection{Continuous Computing Resource Allocation}
In the third phase, the system translates the discrete microservice permutation into continuous resource allocation. Because the data dependency of microservices and resource interference exhibit strong non-convex combinatorial optimization properties, directly solving for continuous rates natively leads to severe local optima. Therefore, we adopt the decoupled scheduling principles common in modern operating systems. The system first determines the execution rank of the microservices. Subsequently, the system partitions the continuous computing resources.

Our system translates the discrete ordering into a continuous relative weight $\tilde{w}_v(t)$. This soft isolation mechanism eliminates the system overhead incurred by hard preemption (e.g., context switching) \cite{fried2020}. To ensure deterministic weight assignment and strictly favor critical microservices, we adopt an exponential decay mapping strategy. For a microservice $v \in \tilde{S}_m(t)$ assigned to the execution rank $r$ (where $r \in \{1, 2, \dots, K\}$), its weight is defined as $\tilde{w}_v(t) = \gamma^{K-r}$. Here, $\gamma > 1$ is a tunable decay factor. For ready microservices that do not fall into the truncated active window, namely $v \in S_m(t) \setminus \tilde{S}_m(t)$, the system assigns a minimal base weight $0 < \epsilon \ll 1$ to prevent complete starvation without disrupting the active window's dominance.

Finally, to obtain a closed-form solution that satisfies multidimensional resource physical constraints under millisecond budgets, we allocate resources based on the classical NUM theory \cite{palomar2007}. We define the computing resource allocation share of microservice $v$ at time $t$ as a continuous scalar $\theta_v(t) \in (0, 1]$. To avoid complex multi-commodity flow calculations under strict millisecond latency budgets, we map this scalar uniformly across all hardware dimensions as a normalized computing resource share. Let $\mathbf{C}_m$ denote the multi-dimensional capacity vector of machine $m$; the execution rate vector of microservice $v$ is formalized as $\mathbf{r}_v^q(t) = \theta_v(t) \mathbf{C}_m$. This implies that the multidimensional physical constraint $\sum_{v \in S_m(t)} \mathbf{r}_v^q(t) \le \mathbf{C}_m$ is mathematically equivalent to a highly simplified scalar constraint $\sum_{v \in S_m(t)} \theta_v(t) \le 1$. Specifically, to achieve weighted proportional fairness, we introduce a logarithmic utility function as the optimization objective:
\begin{equation}
\begin{aligned}
\max_{\boldsymbol{\theta}(t)} \quad & \sum_{v \in S_m(t)} \tilde{w}_v(t) \log(\theta_v(t)) \\
\text{s.t.} \quad & \sum_{v \in S_m(t)} \theta_v(t) \le 1, \quad \theta_v(t) > 0, \; \forall v \in S_m(t)
\end{aligned}
\end{equation}

Because the objective function is strictly concave and the constraints form a standard probability simplex, strong duality holds for this optimization problem. By applying the Karush-Kuhn-Tucker conditions, we deduce that $\theta_v(t) \propto \tilde{w}_v(t)$ and the capacity constraint is tightly binding. This directly yields a closed-form continuous rate allocation scheme:
\begin{equation}
\theta_v^*(t) = \frac{\tilde{w}_v(t)}{\sum_{u \in S_m(t)} \tilde{w}_u(t)} \implies \mathbf{r}_v^q(t) = \left( \frac{\tilde{w}_v(t)}{\sum_{u \in S_m(t)} \tilde{w}_u(t)} \right) C_m
\end{equation}
This mechanism accurately preserves the relative priority structure implied by the discrete permutation via $\theta_v(t) \propto \tilde{w}_v(t)$. Furthermore, the aggregate normalization ensures that the allocation scheme fully utilizes the node capacity without violating the multidimensional strict boundary $\mathcal{P}_m(t)$. This achieves safe and instantaneous continuous resource slicing.

\subsection{Adaptive Robust Execution and Safety Guarantee}

The SQA-produced ranking prior can be unreliable under severe runtime volatility.
To ensure worst-case robustness, we adopt a learning-augmented execution rule with untrusted predictions \cite{purohit2018,lindermayr2025}.
At each decision epoch $t$, we treat the quantum-guided allocator and a robust baseline allocator as two competing experts, producing feasible rate vectors $\mathbf{r}^q(t)$ and $\mathbf{r}^b(t)$, respectively.
We maintain a trust weight $\lambda(t)\in[0,1]$ for the quantum-guided expert, initialized as $\lambda(1)=0.5$.

\subsubsection{Shadow Loss and Weight Update}

To obtain a fast and stable feedback signal within millisecond budgets, we use a lightweight shadow model to compute per-epoch surrogate losses $\hat{L}^q(t)$ and $\hat{L}^b(t)$ for the two experts.
Specifically, $\hat{L}^q(t)$ (resp.\ $\hat{L}^b(t)$) evaluates only the delay-pressure component $\sum_{v\in S_m(t)}\phi(q_v,\theta_v)$ of the surrogate objective defined in Section~3.2 at the rate vector $\mathbf{r}^q(t)$ (resp.\ $\mathbf{r}^b(t)$). Because the permutation is fixed in Phase~2, the pairwise contention term $\psi$ has been absorbed into the discrete ranking result and does not enter the shadow loss. Each epoch's shadow loss is normalized to $[0,1]$ by dividing by the maximum single-epoch delay pressure $\max_t \sum_{v} q_v(t)/\theta_v(t)$.

We update the trust weight using the standard exponential-weights (Hedge) rule:
\begin{equation}\label{eq:hedge_update}
\lambda(t+1)=\frac{\lambda(t)\exp\!\bigl(-\eta \hat{L}^q(t)\bigr)}{\lambda(t)\exp\!\bigl(-\eta \hat{L}^q(t)\bigr)+(1-\lambda(t))\exp\!\bigl(-\eta \hat{L}^b(t)\bigr)},
\end{equation}
where $\eta>0$ is a learning rate.

\subsubsection{Deterministic Rate Mixing}

The node-level scheduler executes a deterministic convex combination of the two feasible allocations:
\begin{equation}\label{eq:rate_mix}
\mathbf{r}^*(t)=\lambda(t)\,\mathbf{r}^q(t)+\bigl(1-\lambda(t)\bigr)\,\mathbf{r}^b(t).
\end{equation}
Because both $\mathbf{r}^q(t)$ and $\mathbf{r}^b(t)$ lie in the feasibility polytope $\mathcal{P}_m(t)$ and $\mathcal{P}_m(t)$ is convex, the mixed rate $\mathbf{r}^*(t)$ is also feasible for every $\lambda(t)\in[0,1]$.

\subsubsection{Regret Guarantee}

We now state the formal performance guarantee of the adaptive mechanism.

\begin{assumption}\label{as:bounded_loss}
The per-epoch surrogate losses satisfy $\hat{L}^q(t),\,\hat{L}^b(t)\in[0,1]$ for all $t$.
\end{assumption}

\begin{assumption}\label{as:convexity}
The per-epoch scheduling cost $c(\mathbf{r},t)$ is convex in the rate vector $\mathbf{r}$ for each decision epoch $t$.
\end{assumption}

Because the shadow loss is computed from the delay-pressure component $\phi$ alone (see above), Assumption~\ref{as:convexity} requires only that $\phi(q_v,\theta_v)=q_v/\theta_v$ is convex in $\theta_v$. For $\theta_v>0$, the second derivative $\mathrm{d}^2\phi/\mathrm{d}\theta_v^2 = 2q_v/\theta_v^3 > 0$, so Assumption~\ref{as:convexity} holds by construction.

\begin{proposition}\label{prop:regret}
Under Assumptions~\ref{as:bounded_loss} and~\ref{as:convexity}, the cumulative scheduling cost of the mixed policy $\mathbf{r}^*(t)$ over a horizon of $T$ decision epochs satisfies
\begin{equation}\label{eq:regret_bound}
\sum_{t=1}^{T} c\!\bigl(\mathbf{r}^*(t),t\bigr)
\;\le\;
\min\!\Biggl\{\sum_{t=1}^{T} c\!\bigl(\mathbf{r}^q(t),t\bigr),\;\sum_{t=1}^{T} c\!\bigl(\mathbf{r}^b(t),t\bigr)\Biggr\}
\;+\;\frac{\ln 2}{\eta}+\frac{\eta\, T}{8}\,.
\end{equation}
\end{proposition}

\begin{proof}
Fix an arbitrary epoch $t$.
By Assumption~\ref{as:convexity} and Jensen's inequality,
\[
c\!\bigl(\mathbf{r}^*(t),t\bigr)
= c\!\bigl(\lambda(t)\,\mathbf{r}^q(t)+(1{-}\lambda(t))\,\mathbf{r}^b(t),\,t\bigr)
\le \lambda(t)\,c\!\bigl(\mathbf{r}^q(t),t\bigr)+(1{-}\lambda(t))\,c\!\bigl(\mathbf{r}^b(t),t\bigr).
\]
The right-hand side is the cost of a randomized choice between two experts with mixture weight $\lambda(t)$, which is exactly the quantity governed by the Hedge algorithm.
Summing over $t=1,\dots,T$ and applying the two-expert Hedge regret bound \cite[Theorem~2.2]{cesa2006} yields~\eqref{eq:regret_bound}.
\end{proof}

The regret term $\frac{\ln 2}{\eta}+\frac{\eta T}{8}$ is minimized at $\eta^*=\sqrt{8\ln 2/T}$, yielding $O(\sqrt{T})$ cumulative regret and thus $O(1/\sqrt{T})$ vanishing average regret.
In practice, the decision horizon $T$ is not known in advance. We therefore set $\eta$ as a fixed constant calibrated to the expected episode length; Section~5 validates empirically that this choice provides effective adaptation across both shock and recovery regimes.

When the quantum-guided prior becomes unreliable, the surrogate loss $\hat{L}^q(t)$ increases relative to $\hat{L}^b(t)$, causing $\lambda(t)$ to decrease via~\eqref{eq:hedge_update}.
The scheduler thereby shifts weight toward the robust baseline without operator intervention.
Proposition~\ref{prop:regret} guarantees that the cost of this mixed policy never exceeds the cost of the better expert by more than $O(\sqrt{T})$ over any horizon, providing a worst-case safeguard against non-stationary prediction failures.

\section{Result Analysis}

To validate the effectiveness of Q-GARS, we constructed a discrete-event simulation platform. To satisfy the millisecond decision budgets established in Section 3, we deployed the evaluation environment on a high-performance computing cluster equipped with a local NVIDIA A40 GPU. This local coprocessor architecture eliminates cloud communication latency. For the active window size $K \le 100$ defined in Section 3, the resulting QUBO models contain up to $10{,}000$ binary variables. The local SQA solver resolves these models in sub-milliseconds. This satisfies the physical latency constraints of microservice scheduling. The experiments simulated tens of thousands of concurrent workflows using Monte Carlo methods to ensure statistical significance. The evaluation consists of three progressive phases:
\begin{enumerate}
    \item Structural Prior Gain: Quantifying the reduction in the surrogate objective brought by discrete sequence optimization under ideal predictive conditions.
    \item Worst-case Robustness: Evaluating the tail-risk control of the learning-augmented mechanism under varying stochastic noise levels.
    \item System Dynamics at Scale: Verifying long-term stability and automated recovery capabilities when the system faces simulated telemetry failures.
\end{enumerate}

\subsection{Performance Gain from Structural Prior}

First, we quantify the performance improvement brought about by the quantum-inspired structural prior. We generate a data set of directed acyclic graphs. The number of nodes per graph follows a uniform distribution between 10 and 50, with a maximum concurrency width ranging from 2 to 8. We compared the sequence generated by Q-GARS against a greedy baseline policy based on the shortest remaining processing time. To ensure a rigorous comparison, both discrete ordering results are fed into the NUM module described in Section 3 for continuous rate slicing.

Fig. \ref{fig:opt_gains}(a) shows the distribution of optimization gains among $N=10000$ independent samples. The structural prior significantly reduces the global weighted completion time, achieving a 2.1\% mean improvement, with the maximum improvement reaching 16.8\%. The scatter plot in Fig. \ref{fig:opt_gains}(b) further compares the solution quality. As the baseline objective value increases, the solution of Q-GARS consistently remains in the "Better Region" below the diagonal. This demonstrates the solving consistency of the algorithm within complex state spaces.

\begin{figure}[htbp]
    \centering
    \includegraphics[width=\linewidth]{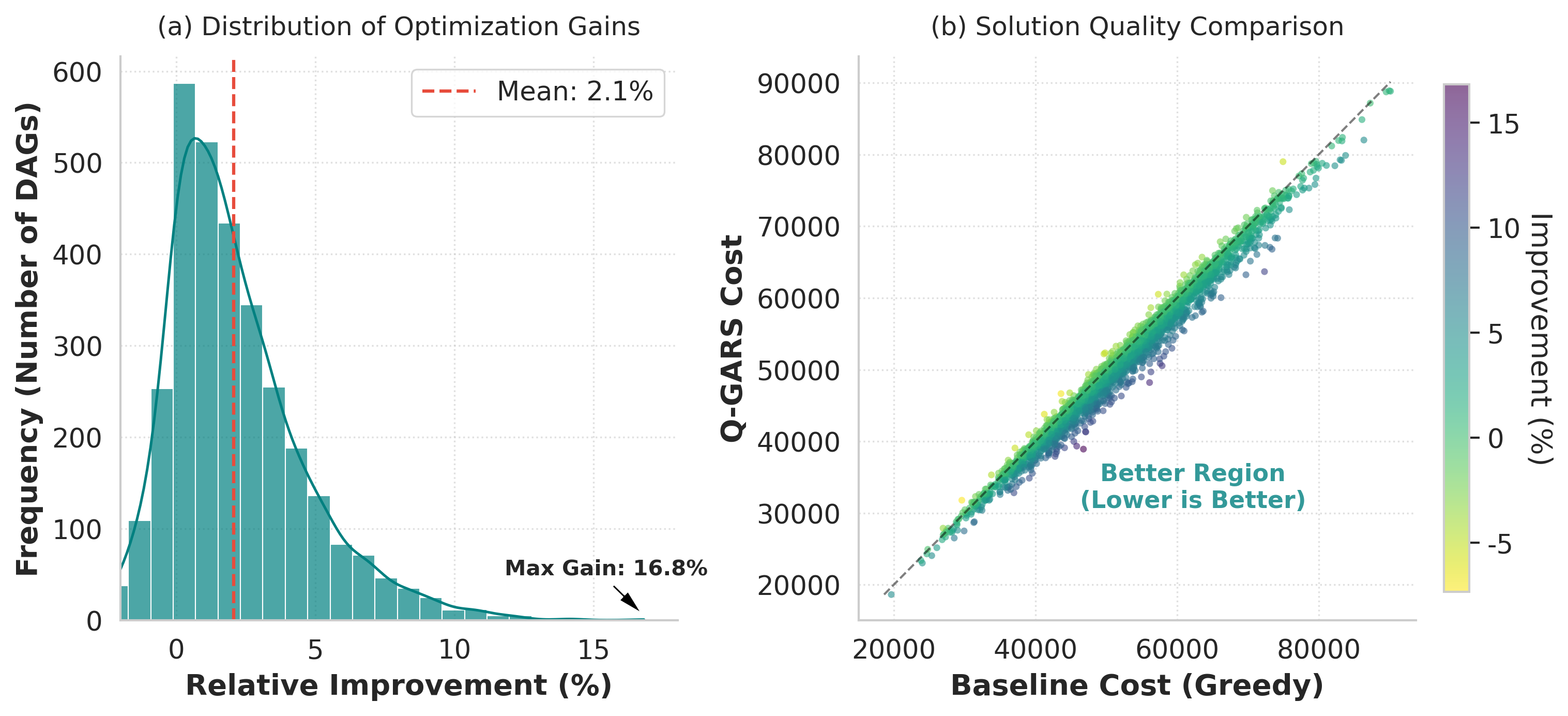}
    \caption{(a) Histogram showing the distribution of relative performance improvements over the greedy baseline. (b) Scatter plot comparing solution quality, indicating consistent improvements in high-cost regions.}
    \label{fig:opt_gains}
\end{figure}

\subsection{Resilience Under Stochastic Volatility}

Real-world cloud-edge environments are fraught with stochastic execution times. In the second phase, we parameterize the stochastic latency model using $\alpha$ to control the magnitude of fluctuations and inject tail latency. Fig. \ref{fig:resilience} visually demonstrates the performance of different scheduling mechanisms regarding the weighted completion time under uncertainty.

As shown in Fig. \ref{fig:resilience}(a), Q-GARS consistently maintains the lowest average weighted completion time at all uncertainty levels. In the near-deterministic regime ($\alpha \approx 0$), the static prior policy performs comparable to Q-GARS. As volatility increases (e.g. $\alpha \ge 0.25$), Q-GARS quickly establishes a gap and consistently outperforms both the static prior and the robust baseline policy. Even in the high-noise regime, the performance of the robust baseline merely approaches Q-GARS, but does not surpass it in our experiments. In a highly volatile setting ($\alpha=1.5$), the upper-tail cumulative probability distribution in Fig. \ref{fig:resilience}(b) shows that Q-GARS and the robust baseline significantly shift the worst completion times to the left compared to the static prior. This indicates that the adaptive trust parameter effectively mitigates extreme latency events. Overall, these observations are consistent with the guarantee of Proposition~\ref{prop:regret}. In the low-uncertainty regime, the ranking prior is accurate and Q-GARS tracks the static prior. In the high-uncertainty regime, the prior degrades, $\lambda(t)$ decreases via the Hedge update, and Q-GARS automatically falls back toward the robust baseline. In the intermediate regime, the mixed policy tracks whichever expert performs better. Proposition~\ref{prop:regret} guarantees that the cumulative cost of the mixed policy exceeds that of the better expert by at most $O(\sqrt{T})$; the crossover behavior of the three curves in Fig.~4(a) is an empirical manifestation of this theoretical property.

\begin{figure}[h]
    \centering
    \includegraphics[width=\linewidth]{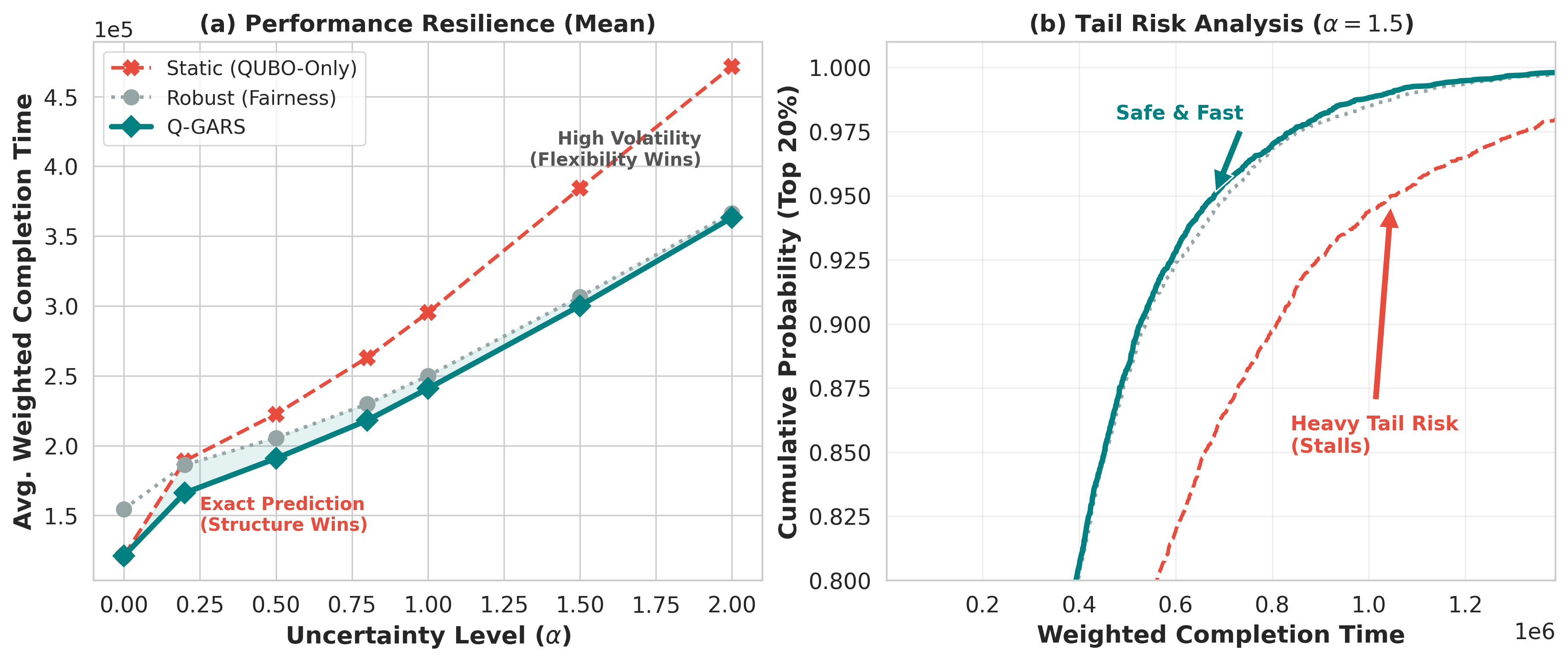}
    \caption{(a) Impact of increasing uncertainty level $\alpha$ on average weighted completion time. (b) Cumulative Distribution Function (CDF) of completion times under high volatility ($\alpha=1.5$), highlighting the mitigation of heavy tail risks.}
    \label{fig:resilience}
\end{figure}

\subsection{System Dynamics and Shock Recovery}

Finally, we evaluated large-scale system resilience through long-horizon simulations of 4096 independent parallel systems. We mark a shock interval from time step 300 to 900 with vertical dashed lines in Fig. \ref{fig:dynamics}. During this period, node failure rates and prediction errors increase sharply.

As shown in Fig. \ref{fig:dynamics}(a), the mean trust parameter between systems $\bar{\lambda}(t)$ drops rapidly during shock. This indicates that the shadow loss signal detects prediction degradation. The system moves toward safeguard mode via the exponential weight update rule detailed in Section 4. After the shock subsides, the trust parameter gradually recovers. The system re-enables the benefits of predictive priors.

Fig. \ref{fig:dynamics}(b) reports the 95th percentile queue backlog. Compared with non-adaptive baselines, Q-GARS significantly suppresses the peak backlog surge. Specifically, its peak is approximately 20\% to 30\% of the SRPT-pred baseline. Fig. \ref{fig:dynamics}(c) shows the blocked capacity ratio. Q-GARS maintains a lower capacity loss during shock interval and returns to near-zero levels fastest after recovery. Together, these results empirically demonstrate the graceful degradation and fast recovery capabilities of the framework in non-stationary environments.

\begin{figure}[htbp]
    \centering
    \includegraphics[width=\linewidth]{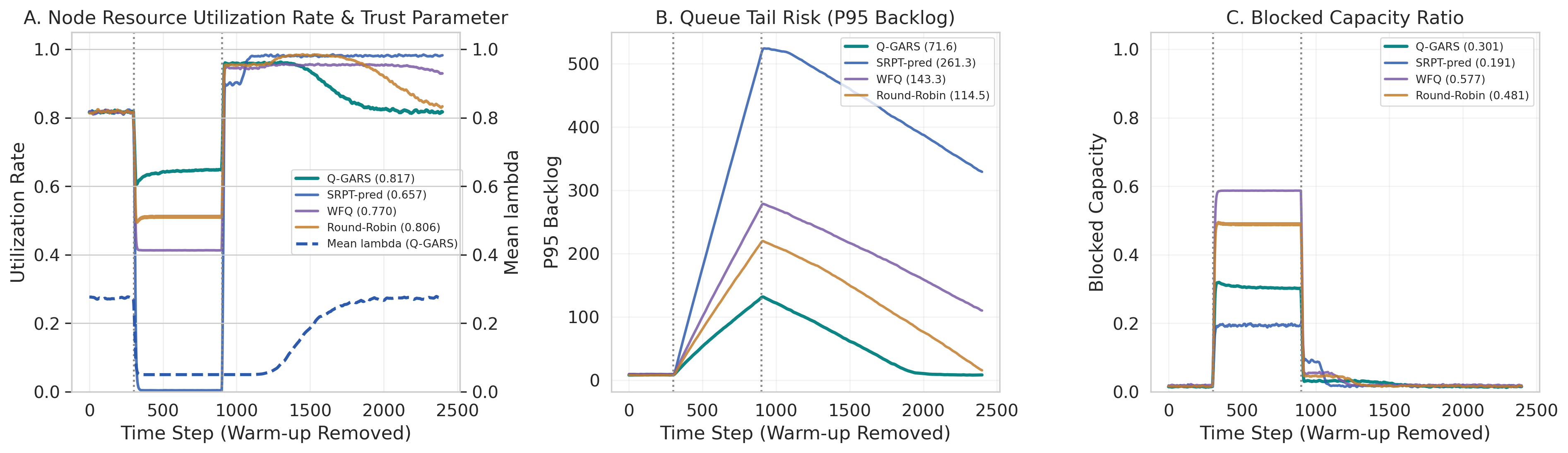}
    \caption{System dynamics under a simulated shock interval from time step 300 to 900. (a) The cross-system mean trust parameter $\bar{\lambda}(t)$ drops to trigger the safeguard mode. (b) Q-GARS effectively suppresses the peak surge of the 95th percentile queue backlog. (c) Blocked capacity ratio is minimized and recovers to near-zero levels fastest.}
    \label{fig:dynamics}
\end{figure}

\section{Conclusion}
This paper proposes a quantum-guided adaptive scheduling framework for the cloud-edge continuum, effectively mitigating head-of-line blocking in heterogeneous computing environments by integrating the combinatorial optimization capabilities of QUBO with the robustness of proportional fairness. Although applying quantum computing in the current Noisy Intermediate-Scale Quantum era faces physical limitations such as restricted qubit connectivity, limited coherence time, and high environmental stochasticity that can lead to diminishing returns and compromise global optimality, our dual-track mechanism successfully establishes a solver-agnostic interface between a combinatorial optimizer and the stochastic physical environment. As a bridge between quantum algorithms and classical real-time control systems, the adaptive control paradigm established in this framework will translate into potential latency advantages once future quantum hardware transcends current scaling limits.

\bibliographystyle{splncs04}
\bibliography{online_opt}
\end{document}